\documentclass[conference]{IEEEtran}

\usepackage[colorlinks=true,pdfstartview=FitV,linkcolor=black,citecolor=black,urlcolor=blue,plainpages=false]{hyperref}

\usepackage[numbers,sort&compress]{natbib}

\usepackage[T1]{fontenc}
\usepackage[latin1]{inputenc}

\usepackage{amsmath}
\usepackage{amsthm}
\usepackage{amsfonts}
\usepackage{dsfont}
\usepackage{mathrsfs}
\usepackage{enumerate}

\usepackage{tikz}

\usepackage[]{authblk}
\usepackage{times}
\usepackage{epsf}
\usepackage{psfrag}
\usepackage{epsfig}

\usepackage{amssymb}
\usepackage{amstext}
\usepackage{latexsym}
\usepackage{color}
\usepackage{ifthen}
\usepackage{multirow}

\usepackage{subfigure}

\renewcommand{\vec}[1]{{\bf #1}}

\newcommand{\ep}{\epsilon}
\newcommand{\vep}{\varepsilon}

\newcommand{\N}{{\mathcal N}}

\newcommand{\R}{{\mathbb R}}

\newcommand{\one}{{\mathds 1}}

\newcommand{\til}{\tilde}

\newcommand{\goesto}{\rightarrow}

\newcounter{constcount}
\newcounter{numcount}
\newcommand{\leqnum}{\stackrel{(\roman{numcount})}{\leq\;}\stepcounter{numcount}}

\newcommand{\cnt}{$(\roman{numcount}$)\;\stepcounter{numcount}}

\newcommand{\rescnt}{\setcounter{numcount}{1}}

\newcommand{\DFT}{\text{DFT}}

\newcounter{thmcnt}
  \let\Oldsection\section
  
\renewcommand{\section}{\stepcounter{thmcnt}\Oldsection}

\newtheorem{theorem}{Theorem}

\newtheorem{lemma}{Lemma}
\newtheorem{definition}{Definition}

\newcounter{examplecounter}
\newcommand{\aln}[1]{\begin{align*}#1\end{align*}}

\newcommand{\al}[1]{\begin{align}#1\end{align}}

\begin{document}


\title{Worst-Case Source for Distributed \\ Compression with Quadratic Distortion}

\author{Ilan Shomorony$^{\dagger}$, A. Salman Avestimehr$^{\dagger}$, Himanshu Asnani$^\ast$ and Tsachy Weissman$^\ast$ \\
$^\dagger$Cornell University, Ithaca, NY\\
$^\ast$Stanford University, Stanford, CA}



\maketitle

%
%
%
%
%

\begin{abstract}
We consider the $k$-encoder source coding problem with a quadratic distortion measure.
We show that among all source distributions with a given covariance matrix $\vec K$, the jointly Gaussian source requires the highest rates in order to meet a given set of distortion constraints.
\end{abstract}

\section{Introduction}

The characterization of the rate-distortion region for the $k$-encoder source coding problem, depicted in Figure \ref{problemfig}, is one of the central open problems in network information theory.
In this problem, $k$ encoders observe different components of a random vector-valued source.
Then, without cooperating, the encoders transmit messages over rate-constrained, noiseless channels to a central decoder, which, based on the $k$ received messages, tries to reproduce the original source.
The goal is to determine which rate tuples $(R_1,...,R_k)$ allow the decoder to reproduce the source so that distortion constraints placed on each of the $k$ components are satisfied.

Most of the work on this problem has focused on the case $k=2$, and, for some specific distortion constraints, the rate-distortion region has been completely characterized.
When both sources must be reconstructed losslessly, we have the classical Slepian-Wolf problem \cite{SlepianWolf}.
When one of the two sources is available to the decoder as side-information, the rate-distortion region was characterized in \cite{AhlswedeKorner,WynerSide,WynerZiv} under different distortion constraints.
The case where one of the sources must be reconstructed losslessly while the other must satisfy an arbitrary distortion constraint was solved by Berger and Yeung \cite{BergerYeung}, and generalizes all the previous cases.

In \cite{Aaron2Terminal}, the rate-distortion region for the two-encoder source coding problem with quadratic distortion constraints and Gaussian sources was completely characterized. 
A by-product of this result was the characterization of the Gaussian source as the worst-case source for the two-encoder quadratic source coding problem, generalizing the well known fact that the Gaussian source has the largest rate-distortion function for a given variance \cite[Example 9.7]{GallagerBook}.

The importance of characterizing the Gaussian source as the worst-case source is two-fold. 
First, it justifies the study of distributed source coding problems for Gaussian sources as a way of obtaining a worst-case analysis for more practical data source models.
The second important aspect is to establish the existence of optimal codes for Gaussian sources which are robust to changes in the source distribution, i.e., they have the same performance guarantees if the sources are non-Gaussian. 

%

However, for the general k-encoder quadratic Gaussian source coding problem, it is still unknown whether the jointly Gaussian source is the worst-case source. 
The proof that the jointly Gaussian sources are the worst-case sources for the two-encoder problem in \cite{Aaron2Terminal} follows from the fact that the Berger-Tung separation-based architecture \cite{BergerMultiterminal,TungMultiterminal} is shown to be optimal for jointly Gaussian sources, and this architecture can achieve the same rate region for any source distribution with a given covariance matrix $\vec K$.
Since this separation-based architecture is not known to be optimal for the general $k$-encoder problem, the same arguments cannot be extended to the general case.
Furthermore, it is in general unclear what kind of performance guarantees can be obtained when codes designed for the $k$-encoder source coding problem with Gaussian sources are employed with non-Gaussian sources. 
Therefore, in order to address these problems, new techniques must be introduced.

\begin{figure}[ht] 
     \centering
       \includegraphics[width=\linewidth]{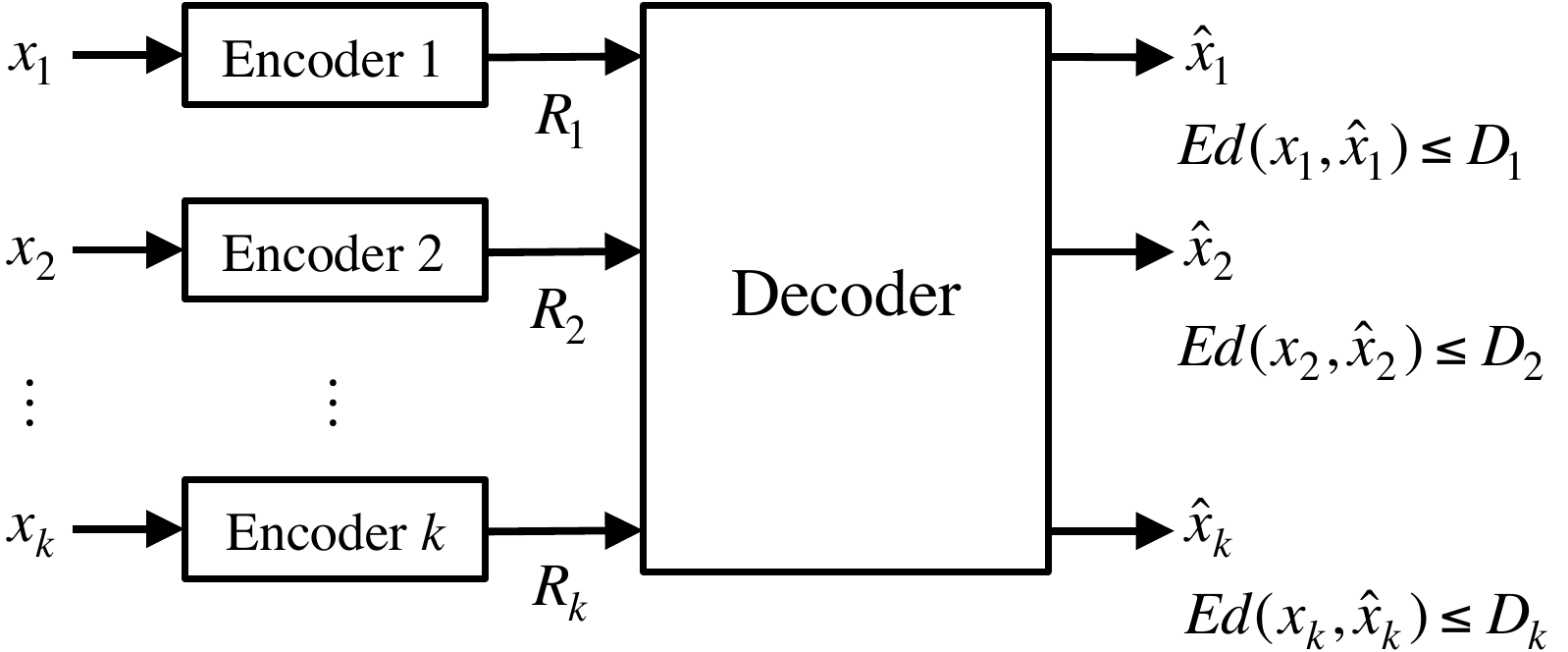} \caption{The $k$-encoder source coding problem. \label{problemfig}}
\end{figure}

Recently, it was shown in \cite{wcnoisefull} that the Gaussian noise is the worst-case noise for general multi-hop multi-flow wireless networks.
The main idea was to apply an OFDM-like scheme at all transmitters and receivers in the network in order to ``mix'' different noise realizations over time.
This mixing, if performed over sufficiently long blocks, allows the Central Limit Theorem to kick in, effectively creating a new network where the additive noises are approximately Gaussian.
This allows a coding scheme designed for a wireless network with Gaussian noise terms to achieve the same rates of reliable communication on a network with non-Gaussian noises.

In this paper, we show that similar ideas to the ones used in \cite{wcnoisefull} can be used in the quadratic $k$-encoder source coding problem, if the source is not Gaussian.
By having each encoder apply a DFT-based unitary linear transformation to its vector of source symbols, it is possible to create an approximately Gaussian source with the same covariance matrix.
This allows us to prove that, for a given covariance matrix, the jointly Gaussian source is indeed the worst-case source for the $k$-encoder source coding problem.
Moreover, this technique can be seen as a way of modifying codes designed for Gaussian sources so that they can be applied to non-Gaussian sources and still have a performance guarantee.

%
%
%

\section{Problem Setup and Main Result}

We consider the $k$-encoder rate-distortion problem with a quadratic distortion measure.
In this problem, $k$ encoders observe different components of a vector-valued i.i.d.~sequence $\{(x_1[i],...,x_k[i])\}_{i=0}^{n-1}$.
We assume that $(x_1[0],...,x_k[0])$ has an arbitrary distribution with zero mean and covariance matrix $\vec K$.
Encoder $m$ maps $\vec x_m = (x_m[0],...,x_m[n-1])$ to an integer $f_m(\vec x_m) \in \{1,...,2^{nR_m}\}$, which is transmitted noiselessly to a central decoder.
Given the $k$ integers $f_m(\vec x_m)$, $m=1,...,k$, the decoder uses decoding functions $g_1,...,g_m$ in order to obtain estimates $(\hat x_m[0],...,\hat x_m[n-1]) = g_m \left(f_1(\vec x_1),...,f_k(\vec x_k)\right)$, for $m=1,...,k$.
A code for the $k$-Encoder Rate-Distortion problem is comprised of a set of encoding and decoding functions $\left(f_1,...,f_m, g_1,...,g_k\right)$ for a given blocklength $n$.


\begin{definition}
Rate-distortion vector $(R_1,...,R_k,D_1,...,D_k)$ is achievable if, for some blocklength $n$, there exists a code $(f_1,...,f_m,g_1,...,g_m)$ for which
\al{ \label{distconstraint}
\tfrac1n E\left[\left\| \vec x_m-g_m(f_1(\vec x_1),...,f_k(\vec x_k)) \right\|^2\right] \leq D_m,
}
for $m=1,...,k$.
 \end{definition}
      
The following result establishes that the jointly Gaussian distribution is the worst-case source distribution among those with covariance matrix $\bf K$. 

\begin{theorem} \label{mainthm}
If rate-distortion vector $(R_1,...,R_k,D_1,...,D_k)$ is achievable when $(x_1[0],...,x_k[0])$ is jointly Gaussian with covariance matrix $\vec K$, then, for any $\ep > 0$, rate-distortion vector $(R_1+\ep,...,R_k + \ep,D_1 +\ep,...,D_k+\ep)$ is achievable when $(x_1[0],...,x_k[0])$ has any arbitrary distribution with covariance matrix $\vec K$.
\end{theorem}



\section{Proof of Main Result}

In order to prove Theorem \ref{mainthm}, we will need the following lemma, whose proof is in the Appendix.

\begin{lemma} \label{contlemma}
Assume $(x_1[0],...,x_k[0])$ is jointly Gaussian.
For any code $(f_1,...,f_k,g_1,...,g_k)$ that achieves rate-distortion vector $(R_1,...,R_k,D_1,...,D_k)$ and any $\ep, \ep' > 0$, one can find another code $(\tilde f_1,...,\tilde f_k,\tilde g_1,...,\tilde g_k)$ that achieves the  rate-distortion vector $(R_1+\ep,...,R_k + \ep,D_1 +\ep',...,D_k+\ep')$ 
for which the set of discontinuities of each $\tilde f_m$, $m=1,...,k$, has Lebesgue measure zero.
\end{lemma}


\begin{proof}[Proof of Theorem \ref{mainthm}]
Suppose the rate-distortion vector $(R_1,...,R_k,D_1,...,D_k)$ is achievable in the case where $(x_1[0],...,x_k[0])$ is jointly Gaussian with covariance matrix $\vec K$.
Fix $\ep > 0$.
From Lemma \ref{contlemma}, we can assume that we have a code $(f_1,...,f_k,g_1,...,g_k)$ with blocklength $n$, which 
achieves rate-distortion vector $(R_1+\ep,...,R_k + \ep,D_1 +\ep/2,...,D_k+\ep/2)$
if $(x_1[0],...,x_k[0])$ is jointly Gaussian, and such that
the set of discontinuities of each $f_m$, $m=1,...,k$, has Lebesgue measure zero.
We will then construct new encoding functions $\til f_1,...,\til f_k$ with blocklength $nb$, for a large integer $b$, where $\til f_m$ is applied to the source sequence $\vec x_m = (x_m[0],...,x_m[nb-1])$, for $m=1,...,k$.
The construction of these new encoding functions is illustrated in Figure \ref{encodingfig}.
\begin{figure}[ht] 
     \centering
       \includegraphics[width=\linewidth]{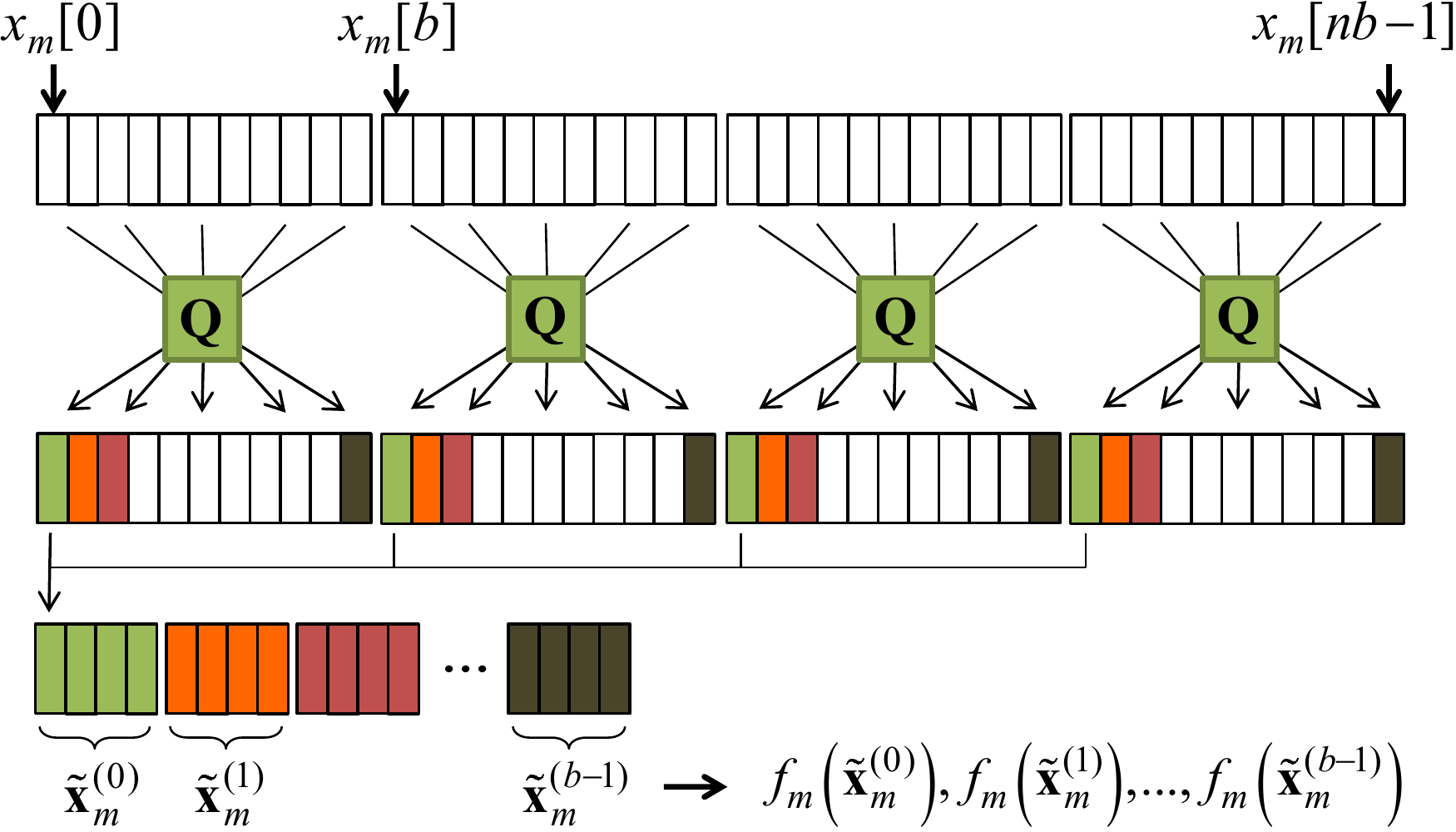} \caption{Illustration of the new encoding procedure for encoder $m$. \label{encodingfig}}
\end{figure}
Encoder $m$ starts by applying a unitary ($2$-norm preserving) linear transformation $\vec Q$ (defined later) to each block of length $b$.
The $n$ resulting blocks of length $b$ are then interleaved, generating $b$ length-$n$ vectors $\til{\vec x}_m^{(0)},...,\til{\vec x}_m^{(b-1)}$, as shown in Figure \ref{encodingfig}.
The original encoding function $f_m$ (which takes as input a length-$n$ vector) is then individually applied to each $\til{\vec x}_m^{(i)}$, for $i=0,...,b-1$.
This generates $b$ integers in $\{1,..., 2^{n(R_m+\ep)}\}$ which can then be combined into a single integer from $\{1,...,2^{nb(R_m+\ep)}\}$ to produce the encoder output $\tilde f_m(\vec x_m)$.

At the decoder side, each $\til f_m(\vec x_m)$, for $m=1,...,k$, is first broken into the $b$ original integers from $\{1,..., 2^{n(R_m+\ep)}\}$.
Then, using the original decoding function $g_m$, the decoder obtains estimates of $\vec{\tilde{x}}_m^{(i)}$, for $i=0,...,b-1$, which can then be converted to an estimate of $\vec x_m$ by applying $\vec Q^{-1}$ $n$ times.
This defines the new decoding functions $\til g_m$, $m=1,...,k$.

We define the unitary matrix $\vec Q$ by having the entry in the $(i+1)$th row and $(j+1)$th column be
\aln{
Q{(i,j)} =\left\{ \begin{array}{ll} 1/\sqrt{b} & \text{if $i = 0$} \\ \sqrt{2/b} \cos\left( \frac{2 \pi j \ell}{b} \right) & \text{if $i = 1,...,\frac b2 - 1$} \\  (-1)^j/\sqrt{b} & \text{if $i = \frac b2 $} \\ \sqrt{2/b} \sin\left( \frac{2 \pi j (\ell-b/2) }{b} \right) & \text{if $i = \frac b2+1,...,b - 1$} \end{array} \right. 
}
for $i,j \in \{0,...,b-1\}$.
We point out that applying the linear transformation $\vec Q$ to a vector $\vec x$ can be seen as first taking the $\DFT$ of $\vec x$, then separating the real and imaginary parts of the resulting vector, and renormalizing them so that the resulting transformation is unitary.
Checking that $\vec Q$ is a unitary transformation, i.e., that $\|\vec Q \vec x\| = \|\vec x\|$ for any $\vec x \in \R^b$, is straightforward and thus omitted.

Our next goal is to show that, by choosing $b$ large enough, we can make the distortion of this new code arbitrarily close to the distortion of the original code applied to the Gaussian source.
We start by noticing that, since $\vec Q$ is a unitary linear transformation, the distortion of our new code can be written in terms of ${\vec {\til x}}_m^{(\ell)}$ for $\ell = 0,...,b-1$ as
\aln{
\frac1{b} \sum_{\ell = 0}^{b-1} \frac1n\left\| \vec{\til x}_m^{(\ell)} - g_m\left( f_1(\vec{\til x}_1^{(\ell)}),...,f_k(\vec{\til x}_k^{(\ell)}) \right)  \right\|^2.
}
For each $b=1,2,...$, we will let 
\aln{
\ell_b = \arg \max_{0 \leq \ell \leq b-1} E \left\| \vec{\til x}_m^{(\ell)} - g_m\left( f_1(\vec{\til x}_1^{(\ell)}),...,f_k(\vec{\til x}_k^{(\ell)}) \right)  \right\|^2,
}
i.e., the $\ell_b$th length-$n$ block has the largest expected distortion.
Note that $\left\{ \left(\til x_1^{(\ell_b)}[i],...,\til x_k^{(\ell_b)}[i]\right) \right\}_{i=0}^{n-1}$ is an i.i.d.~sequence of length-$k$ random vectors.
We will show that it converges in distribution to a sequence of i.i.d.~jointly Gaussian random vectors with covariance matrix $\vec K$, as $b \to \infty$.
Clearly, it suffices to show that $\left(\til x_1^{(\ell_b)}[0],...,\til x_k^{(\ell_b)}[0]\right)$ converges in distribution to a jointly Gaussian random vector with covariance matrix $\vec K$, as $b \to \infty$.
In order to use the Cram\'er-Wold Theorem, we fix an arbitrary vector $(t_1,...,t_k) \in \R^k$ and we notice that
\al{ 
\sum_{m=1}^k t_m \til x_m^{(\ell_b)}[0] 
& = \sum_{m=1}^k t_m \sum_{j=0}^{b-1} x_m[j] ~ Q(\ell_b,j) \nonumber \\
& = \sum_{j=0}^{b-1} \left( \sum_{m=1}^k t_m x_m[j] \right) Q(\ell_b,j). \label{cramerexp}
}
To characterize the convergence in distribution of (\ref{cramerexp}), we will need the following result.
\begin{theorem}[Lindeberg's Central Limit Theorem \cite{billingsley}] \label{lindthm}
Suppose that for each $b = 1,2,...$, the random variables
$
Y_{b,1}, Y_{b,2},..., Y_{b,b}
$
are independent.
In addition, suppose that, for all $b$ and $i \leq b$, $E[Y_{b,i}] = 0$, and let
\al{
s_b^2 = \sum_{i=1}^b E\left[ Y_{b,i}^2 \right].   \label{sbdef}
} 
Then, if for all $\vep > 0$, Lindeberg's condition
\al{
\frac{1}{s_b^2}\sum_{i=1}^b E\left(Y_{b,i}^2 \, \one\left\{|Y_{b,i}|  \geq \vep s_b\right\}\right) \goesto 0 \text{ as $b \goesto \infty$}
\label{lind}
}
holds, we have that
\aln{
\frac{\sum_{i=1}^b Y_{b,i}}{s_b} \stackrel{d}{\goesto} \N(0,1).
}
\end{theorem}
To apply Lindeberg's CLT, we will let, for $j=0,...,b-1$,
\aln{
Y_{b,j+1} = \sqrt b \left( \sum_{m=1}^k t_m x_m[j] \right) Q(\ell_b,j).
}
Then, if we let $\vec K_{u,v}$ be the entry in the $u$th row and $v$th column of $\vec K$, we have
\aln{
s_b^2 & = \sum_{j=1}^b E\left[ Y_{b,j}^2 \right] = b \sum_{j=1}^b  Q^2(\ell_b,j-1) \\
& \quad \quad \quad \times E \left( \sum_{m=1}^k t_m x_m[j-1] \right)^2 \\
& = b \sum_{1 \leq u,v \leq k} t_u t_v \vec K_{u,v}  \sum_{j=1}^b  Q^2(\ell_b,j-1)  \\
& = b \sum_{1 \leq u,v \leq k} t_u t_v \vec K_{u,v},
} 
regardless of the value of $\ell_b$.
In order to verify Lindeberg's condition, we define $\sigma^2 =\sum_{1 \leq u,v \leq k} t_u t_v \vec K_{u,v}$ and we let $U_{b,j} = Y_{b,j}^2 \, \one \left\{|Y_{b,j}|  \geq \vep s_b \right\} = Y_{b,j}^2 \, \one \left\{|Y_{b,j}|  \geq \vep \sigma \sqrt b \right\}$.
Consider any sequence $j_b$, for $b=1,2,...$, such that $j_b \in \{1,...,b\}$, and any $\delta > 0$.
Then we have that
\aln{
\Pr\left(U_{b,j_b} < \delta\right) & \geq \Pr\left(|Y_{b,j_b}| < \vep \sigma \sqrt{b} \right) \\
& \geq \Pr\left(\left|\sum_{m=1}^k t_m x_m[j_b-1]\right|\sqrt2 < \vep \sigma \sqrt{b} \right) \\
& = \Pr\left(\left|\sum_{m=1}^k t_m x_m[0]\right| < \vep \sigma \sqrt{b/2} \right) \goesto 1, 
}
as $b \goesto \infty$, which means that $U_{b,j_b} \stackrel{p}{\goesto} 0$ as $b \goesto \infty$.
Moreover, we have that 
\aln{
\Pr\left( |U_{b,j_b}| \geq t \right) & \leq \Pr\left[ 2 \left(\sum_{m=1}^k t_m x_m[j_b-1]\right)^2 \geq t\right] \\
& =\Pr\left[ 2 \left(\sum_{m=1}^k t_m x_m[0]\right)^2 \geq t\right]
}
for $t > 0$, and
\aln{
E \left[ 2 \left(\sum_{m=1}^k t_m x_m[0]\right)^2 \right] = 2 \sigma^2 < \infty,
}
and by the Dominated Convergence Theorem 
\cite[pages 338-339]{billingsley}, we have that $E[U_{b,j_b}]\goesto 0$ as $b \to \infty$.
We conclude that
\aln{
\frac{1}{s_b^2}\sum_{i=1}^b E\left(Y_{b,j}^2 \, \one\left\{|Y_i|  \geq \vep s_b\right\}\right) & = \frac{1}{\sigma^2 b} \sum_{j=1}^b E\left[ U_{b,j} \right] \\
& \leq \frac{1}{\sigma^2} \max_{1 \leq j \leq b} E\left[ U_{b,j} \right]
 \to 0, 
}
as $b \to \infty$, and Lindeberg's condition (\ref{lind}) is satisfied for any $\vep > 0$.
Hence, from Theorem \ref{lindthm}, we have that
\aln{
\frac{\sum_{i=1}^{b} Y_{b,j}}{\sigma \sqrt{b}} \stackrel{d}{\goesto} \N(0,1), 
} 
which implies, from (\ref{cramerexp}), that
\aln{
\sum_{m=1}^k t_m \til x_m^{(\ell_b)}[0] & =  \sum_{j=0}^{b-1} \left( \sum_{m=1}^k t_m x_m[j] \right) Q(\ell_b,j) \\
& = \frac{\sum_{j=1}^{b} Y_{b,j}}{\sqrt{b}} \stackrel{d}{\goesto} \N(0,\sigma^2).
}
Finally, since for a jointly Gaussian vector $(y_1,...,y_k)$ with mean zero and covariance matrix $\vec K$, we have $\sum_{m=1}^k t_m y_m \sim \N(0,\sigma^2)$, we conclude, from the Cram\'er-Wold Theorem that $\left(\til x_1^{(\ell_b)}[0],...,\til x_k^{(\ell_b)}[0]\right)$ converges in distribution to a jointly Gaussian random vector with zero mean and covariance matrix $\vec K$, as $b \to \infty$.

Now, since the set of discontinuities of $f_m$, for $m=1,...,k$, has Lebesgue measure zero, it is easy to see that the mapping
\aln{
\left\{\vec{\til x}_m^{(\ell)} \right\}_{m=1}^k \mapsto \left\| \vec{\til x}_m^{(\ell)} - g_m\left( f_1(\vec{\til x}_1^{(\ell)}),...,f_k(\vec{\til x}_k^{(\ell)}) \right)  \right\|^2,
}
for $m=1,...,k$,
must also have a set of discontinuities with Lebesgue measure zero.
We conclude that 
\aln{
& \left\| \vec{\til x}_m^{(\ell_b)} - g_m\left( f_1(\vec{\til x}_1^{(\ell_b)}),...,f_k(\vec{\til x}_k^{(\ell_b)}) \right)  \right\|^2 \\
& \quad \quad  \stackrel{d}{\to} 
\left\| \vec y_m - g_m\left( f_1(\vec y_1),...,f_k(\vec y_k) \right)  \right\|^2,
}
as $b \to \infty$, 
where $\vec y_m = (y_m[0],...,y_m[n-1])$, for $m=1,...,k$, and $\left\{ \left(y_1[i],...,y_k[i]\right) \right\}_{i=0}^{n-1}$  is an i.i.d.~sequence such that $(y_1[0],...,y_k[0])$ is jointly Gaussian with zero mean and covariance matrix $\vec K$.
Moreover, we have that 
\aln{
& \left\| \vec{\til x}_m^{(\ell_b)} - g_m\left( f_1(\vec{\til x}_1^{(\ell_b)}),...,f_k(\vec{\til x}_k^{(\ell_b)}) \right)  \right\|^2 \\
& \quad \quad \leq 2 \left\| \vec{\til x}_m^{(\ell_b)} \right\|^2 + 2 \left\| g_m\left( f_1(\vec{\til x}_1^{(\ell_b)}),...,f_k(\vec{\til x}_k^{(\ell_b)}) \right)  \right\|^2 \\
& \quad \quad \leq 2 \left\| \vec{\til x}_m^{(\ell_b)} \right\|^2 + 2 \max_{c_1,...,c_k} \left\| g_m(c_1,...,c_k) \right\|^2,
}
and also that
\aln{
E \left\|\vec{\til x}_m^{(\ell_b)} \right\|^2 & = n E \left(  \sum_{j=0}^{b-1} x_m[j] ~ Q(\ell_b,j) \right)^2 \\
& = n \vec K_{m,m} \sum_{j=0}^{b-1} Q^2(\ell_b,j) \\ 
& = n \vec K_{m,m} < \infty.
}
Thus, from a variation of the Dominated Convergence Theorem (see Problem 16.4 in \cite{billingsley}), we conclude that, as $b\to\infty$,
\aln{
& E\left\| \vec{\til x}_m^{(\ell_b)} - g_m\left( f_1(\vec{\til x}_1^{(\ell_b)}),...,f_k(\vec{\til x}_k^{(\ell_b)}) \right)  \right\|^2 \\
& \quad  \to E\left\| \vec y_m - g_m\left( f_1(\vec y_1),...,f_k(\vec y_k) \right)  \right\|^2 \leq n (D_m+\ep/2).
}
Therefore, we can choose $b$ sufficiently large so that 
\aln{
& \frac1n E \left\| \vec{\til x}_m^{(\ell_b)} - g_m\left( f_1(\vec{\til x}_1^{(\ell_b)}),...,f_k(\vec{\til x}_k^{(\ell_b)}) \right)  \right\|^2 \\ 
& \quad \quad \quad \quad \leq \frac1n E \left\| \vec y_m - g_m\left( f_1(\vec y_1),...,f_k(\vec y_k) \right)  \right\|^2 + \ep/2 \\
& \quad \quad \quad \quad \leq D_m + \ep.
}
The expected distortion of our code $(\tilde f_1,...,\tilde f_k,\tilde g_1,...,\tilde g_k)$ (with blocklength $nb$) thus satisfies
\aln{
& \frac1{nb} \sum_{\ell = 0}^{b-1} E \left\| \vec{\til x}_m^{(\ell)} - g_m\left( f_1(\vec{\til x}_1^{(\ell)}),...,f_k(\vec{\til x}_k^{(\ell)}) \right)  \right\|^2 \\ 
& \quad \quad \quad \quad \leq \frac1n E \left\| \vec{\til x}_m^{(\ell_b)} - g_m\left( f_1(\vec{\til x}_1^{(\ell_b)}),...,f_k(\vec{\til x}_k^{(\ell_b)}) \right)  \right\|^2 \\
& \quad \quad \quad \quad \leq D_m + \ep,
}
for $m=1,...,k$.
This concludes the proof of Theorem \ref{mainthm}. 
\end{proof}

\begin{appendix}

\section{Appendix}


\begin{proof}[Proof of Lemma \ref{contlemma}]
Let $(x_1[0],...,x_k[0])$ be jointly Gaussian.
If we assume that the rate-distortion vector $(R_1,...,R_k,D_1,...,D_k)$ is achievable, for some blocklength $n$, there exists a code $(f_1,...,f_m,g_1,...,g_m)$  for which (\ref{distconstraint})
is satisfied for $m=1,...,k$.
We follow the construction from \cite{chensemicontinuity} to build a code $(\til f_1,...,\til f_m,\til g_1,...,\til g_m)$ with the same blocklength $n$, which satisfies
\aln{ 
\tfrac1n E \left\| \vec x_m-\til g_m(\til f_1(\vec x_1),...,\til f_k(\vec x_k)) \right\|^2  \leq D_m + \ep'
}
for $m=1,...,k$.

Since our code can be repeated over multiple blocks of length $n$, we may assume that $n$ is large enough so that $2^{nR_m} + 1 \leq 2^{n(R_m+\ep)}$ for each $m$.
Focus on encoder $f_1$, and let $B_j = f_1^{-1}(j)$, for $j \in \{1,...,2^{nR_1}\}$.
Then, the $B_j$'s are a partition of $\R_n$.
For each $j$, from Theorem 11.4 in \cite{billingsley}, for any $\delta > 0$, there exists a countable (in fact, finite) union of disjoint bounded rectangles $\tilde B_j$ such that $\Pr[\vec x_1 \in B_j \vartriangle \tilde B_j] < \delta$.
Then we define $\tilde f_1$ as
\aln{
\tilde f_1(\vec x_1) = \left\{ \begin{array}{ll} j & \text{if $\vec x \in \tilde B_j \setminus \bigcup_{i \ne j}\tilde B_i$} \\ 0 & \text{otherwise.} \end{array} \right.
}
We create the encoders $\tilde f_2,...,\tilde f_k$ in the same way.
For $m=1,...,k$, our decoders will be
\aln{
\tilde g_m(j_1,...,j_k) = \left\{ \begin{array}{ll} g_m(j_1,...,j_k) & \text{if $j_i \ne 0$ for $i=1,...,k$} \\ 0 & \text{otherwise.} \end{array} \right.
}
The new code is similar to the original one in the sense that, if we let $A$ be the event 
\aln{
& \left\{ g_m\left(f_1(\vec x_1),...,f_k(\vec x_k)\right) \ne \tilde g_m\left(\tilde f_1(\vec x_1),...,\tilde f_k(\vec x_k)\right) \right. \\
& \hspace{20mm} \left. \text{ for some } m \in \{1,...,k\} \right\} } 
then, by the union bound,
\aln{
\Pr\left[A\right] \leq \delta \sum_{m=1}^k 2^{nR_m}.}
It is clear that this new code has rates at most $R_1 + \ep,...,R_k+\ep$.
Following the derivation in \cite{chensemicontinuity}, the distortion for decoder $g_1$ satisfies
\aln{
& E\left[ \left\| \vec x_1 - \tilde g_1 \left(\tilde f_1(\vec x_1),...,\tilde f_k(\vec x_k)\right) \right\|^2 \right] \\ 
& \leq E\left[ \left\| \vec x_1 - \tilde g_1 \left(\tilde f_1(\vec x_1),...,\tilde f_k(\vec x_k)\right) \right\|^2 \one_{A^c} \right] \\ 
& \quad \quad + E\left[ \left\| \vec x_1 - \tilde g_1 \left(\tilde f_1(\vec x_1),...,\tilde f_k(\vec x_k)\right) \right\|^2 \; \one_A \right] \\
& \leqnum E\left[ \left\| \vec x_1 - g_1 \left( f_1(\vec x_1),..., f_k(\vec x_k)\right) \right\|^2  \right]  + M \sqrt \delta \\
& \leq n D_1 + M \sqrt\delta,
} \rescnt
where \cnt follows by using Cauchy-Schwarz to obtain
\aln{
& E\left[ \left\| \vec x_1 - \tilde g_1 \left(\tilde f_1(\vec x_1),...,\tilde f_k(\vec x_k)\right) \right\|^2 \one_A \right] \\ 
& \leq 2 E\left[ \left\| \vec x_1 \right\|^2 \one_A \right] \ + 2 E\left[  \left\| \tilde g_1 \left(\tilde f_1(\vec x_1),...,\tilde f_k(\vec x_k)\right) \right\|^2 \one_A \right] \\
& \leq 2 E \left[ \left\| \vec x_1 \right\|^4 \right]^{1/2} E \left[   \one_A \right]^{1/2}  \\ 
& \quad \quad + 2  \max_{\vec x} \left\| \tilde g_1 \left(\tilde f_1(\vec x_1),...,\tilde f_k(\vec x_k)\right) \right\|^2  E \left[   \one_A \right]^{1/2} \\
& \leq \left( 2 E \left[ \left\| \vec x_1 \right\|^4 \right]^{1/2} + 2  \max_{\vec x} \left\| \tilde g_1 \left(\tilde f_1(\vec x_1),...,\tilde f_k(\vec x_k)\right) \right\|^2 \right) \\ 
& \quad \quad \times \sqrt{\delta \sum_{m=1}^k 2^{nR_m}} \\
& = M \sqrt \delta,
}
where $M$ is a finite number, independent of $\delta$.
Therefore, we can choose $\delta > 0$ sufficiently small so that $M \sqrt\delta \leq n\ep'$, and the distortion of each decoder $g_m$ is at most $D_m + \ep'$.
%
%
Finally, we need to show that the set of discontinuities of each $\tilde f_m$ has measure zero.
If we again focus on $\tilde f_1$, this function partitions $\R_n$ into $\hat B_j = \tilde B_j \setminus \cup_{i\ne j} \tilde B_i$ for $j=1,...,2^{nR_1}$ and $\hat B_0 = \R^n \setminus \cup_j \hat B_j$.
Moreover, since the $\tilde B_j$'s were countable unions of disjoint bounded rectangles, and the class of bounded rectangles forms a semiring \cite{billingsley}, the $\hat B_j$'s are also countable unions of disjoint bounded rectangles.
Therefore, for a given $j$, we can write $\hat B_j = \cup_{i} S_i$, where the $S_i$'s are disjoint bounded rectangles.
Moreover, we can also write $\hat B_j^c = \cup_i T_i$, where the $T_i$'s are disjoint bounded rectangles.
Thus, we have
\aln{
\partial \hat B_j & = \partial \left( \cup_{i} S_i \right) = \R^n - \left( \cup_i S_i \right)^\circ - \left( \cup_i T_i \right)^\circ \\
& = \left(\cup_i S_i\right) \cup   \left( \cup_i T_i \right) - \left( \cup_i S_i \right)^\circ - \left( \cup_i T_i \right)^\circ \\
& \subseteq \left(\cup_i S_i\right) \cup   \left( \cup_i T_i \right) - \left( \cup_i S_i ^\circ\right) - \left( \cup_i T_i^\circ \right) \\
& = \left(\cup_i \left(S_i - S_i^\circ\right) \right) \cup   \left( \cup_i \left(T_i - T_i^\circ \right) \right) \\
& \subseteq \left(\cup_i \partial S_i\right) \cup   \left( \cup_i \partial T_i \right)
}
Since the boundary of a bounded rectangle clearly has Lebesgue measure zero, we have, for each $i$, $\lambda(\partial S_i) = \lambda( \partial T_i) = 0$, and we conclude that
\aln{
\lambda(\partial \hat B_j) \leq \sum_i \lambda \left( \partial S_i \right) +  \sum_i \lambda \left( \partial T_i \right) = 0,
}
implying that the boundary of the partition of $\R_n$ induced by $\tilde f_m$ has Lebesgue measure zero.
\end{proof}

\end{appendix}

\bibliographystyle{unsrt}

\bibliography{refs}

\end{document}